\PassOptionsToPackage{bookmarks={false}}{hyperref}
\documentclass[conference]{IEEEtran}
\IEEEoverridecommandlockouts
\usepackage{cite}
\usepackage{amsmath,amssymb,amsfonts,amsthm}

\usepackage{bbm}
\theoremstyle{plain}
\newtheorem{theorem}{\bf{Theorem}}
\newtheorem{lemma}{\bf{Lemma}}
\newtheorem{assumption}{\bf{Assumption}}
\theoremstyle{remark}
\newtheorem{remark}{\bf{Remark}}
\newenvironment{proofsketch}{%
  \proof}{\endproof}
\usepackage[ruled,vlined]{algorithm2e}
\usepackage{makecell} % add
\usepackage{graphicx}
\usepackage{subfigure}
\usepackage{textcomp}
\usepackage{xcolor}
\usepackage{bm}
\usepackage{multirow}
\usepackage{setspace}
\usepackage[normal]{caption}
\renewcommand{\raggedright}{\leftskip=0pt \rightskip=0pt plus 0cm} % so amazing!!!
\usepackage{enumitem} % item indent
\usepackage{hyperref}
\usepackage[hyphenbreaks]{breakurl}

\def\BibTeX{{\rm B\kern-.05em{\sc i\kern-.025em b}\kern-.08em
    T\kern-.1667em\lower.7ex\hbox{E}\kern-.125emX}}
\begin{document}

\title{Asynchronous Semi-Decentralized Federated Edge Learning for Heterogeneous Clients}
\author{\IEEEauthorblockN{Yuchang Sun$^*$, Jiawei Shao$^*$, Yuyi Mao$^\dag$, and Jun Zhang$^*$}
\IEEEauthorblockA{$^*$Dept. of ECE, The Hong Kong University of Science and Technology, Hong Kong\\
$^\dag$Dept. of EIE, The Hong Kong Polytechnic University, Hong Kong \\
Email: \{yuchang.sun, jiawei.shao\}@connect.ust.hk,
yuyi-eie.mao@polyu.edu.hk,
eejzhang@ust.hk}
}
\maketitle

\begin{abstract}
Federated edge learning (FEEL) has drawn much attention as a privacy-preserving distributed learning framework for mobile edge networks. In this work, we investigate a novel semi-decentralized FEEL (SD-FEEL) architecture where multiple edge servers collaborate to incorporate more data from edge devices in training. Despite the low training latency enabled by fast edge aggregation, the device heterogeneity in computational resources deteriorates the efficiency. This paper proposes an asynchronous training algorithm for SD-FEEL to overcome this issue, where edge servers can independently set deadlines for the associated client nodes and trigger the model aggregation. To deal with different levels of staleness, we design a staleness-aware aggregation scheme and analyze its convergence performance. Simulation results demonstrate the effectiveness of our proposed algorithm in achieving faster convergence and better learning performance.
\end{abstract}


\section{Introduction}
The rapid developments of artificial intelligence (AI) and Internet of Things (IoT) are boosting the revolution of multiple industries such as intelligent manufacturing, smart healthcare, and home automation. The driving force behind the success of AI comes from the massive data generated by pervasive mobile and IoT devices \cite{verma2017survey}. However, most sensory data contain privacy-sensitive information, and as a consequence, offloading them to the Cloud for centralized data analytics will breach the data privacy requirements \cite{frustaci2017evaluating}. In 2017, Google proposed a privacy-preserving distributed machine learning (ML) framework, namely \emph{federated learning} (FL) \cite{mcmahan2017communication}, where the client nodes (e.g., mobile and IoT devices) are coordinated by a Cloud-based parameter server (PS) to collaboratively train ML models without disclosing their local data. To reduce the training latency, FL was also fused with the emerging paradigm of \emph{mobile edge computing} (MEC) \cite{mao2017survey}, creating a new FL architecture dubbed \emph{federated edge learning} (FEEL) \cite{lim2020federated} that keeps the entire training process within the edge of wireless networks.

Prior studies on FEEL were largely restricted to synchronous training, where the client nodes perform the same amount of local training and upload their model updates to an edge server simultaneously for model aggregation \cite{wang2019adaptive,chen2020convergence}. Nevertheless, in many real-life use cases of FEEL, the client nodes, e.g., unmanned vehicles, smartphones, and wearable devices, may be of diversified computational resources, including processing speed, battery capacity, and memory usage \cite{lim2020federated}. Therefore, it may take a longer time for client nodes with less capable computation strength to complete their local training and make them become stragglers for model uploading, which unnecessarily idles the fast client nodes and thereby slows down the training process. Notably, the drawback of synchronous FEEL becomes more prominent for large-scale implementations \cite{yang2021characterizing,bonawitztowards}. 

To accelerate the training process, asynchronous FL has gained increasing attention for its advantages in dealing with device heterogeneity \cite{xie2019asynchronous,ma2021fedsa,wu2020safa}. It allows each client node to upload the model updates independently and the PS to aggregate the received models in an event-trigger fashion. The first asynchronous training algorithm for FL was proposed in \cite{xie2019asynchronous}, where the PS performs model aggregation whenever it receives an update from the client nodes. In this way, model aggregation takes place more frequently at a higher communication cost. To strike a balance between model improvement and latency in each training round, a semi-asynchronous FL mechanism was developed in \cite{ma2021fedsa}, where the PS delays the model aggregation until sufficient local model updates are collected. Nevertheless, asynchronous training may degrade the learning performance since the stale local models from the straggling client nodes may be poisonous for the global model aggregation. This issue was relieved by the design in \cite{wu2020safa}, which allows most of the client nodes to stay asynchronous, whereas those with up-to-date or deprecated local models are forced to synchronize with the edge server. However, when being implemented over wireless networks, the benefits of asynchronous FL are bottlenecked by the limited coverage of a single edge server, which makes it difficult to accommodate a large number of client nodes and leverage their local data.

Recent works considered deploying multiple edge servers to engage more client nodes in the training process of FEEL so that massive distributed data resources can be utilized. In particular, a client-edge-cloud hierarchical FEEL system was proposed in \cite{liu2020client}, where the edge servers collect model updates from the client nodes with lower communication latency, and the Cloud-based PS performs global model aggregation to involve more training data. Nevertheless, even with infrequent communication with the Cloud, hierarchical FEEL still incurs too high latency, especially when millions of model parameters need to be transferred. To exempt communication with the Cloud, a novel FEEL architecture, namely \emph{semi-decentralized federated edge learning} (SD-FEEL) was investigated in \cite{sun2021semi,castiglia2020multi}, which enables fast model exchange among edge servers in replacement of global model aggregation at the Cloud. However, preliminary studies on SD-FEEL focused on synchronous training, which still suffers from low efficiency due to the straggler effect caused by the heterogeneous devices.

In this paper, we propose an asynchronous training algorithm for SD-FEEL, where each edge server sets a deadline to collect model updates from its associated client nodes for \emph{intra-cluster model aggregation}. Then, the intra-cluster aggregated model is shared with the neighboring edge servers for \emph{inter-cluster model aggregation}. In this way, client nodes can perform local training continuously without having to wait for the stragglers. As a result, the local computational resources can be fully exploited and the efficiency of training in SD-FEEL is thus optimized. We prove the convergence of the proposed asynchronous training algorithm, and our theoretical analysis also demonstrates how the convergence rate of asynchronous SD-FEEL is degraded by device heterogeneity. Simulation results corroborate our analysis and show that asynchronous SD-FEEL converges to a model with better accuracy and faster than synchronous SD-FEEL.

The organization of this paper is as follows. In Section \ref{sec:system}, we introduce the system model of SD-FEEL considering heterogeneous devices. In Section \ref{sec:training}, we develop an asynchronous training algorithm for SD-FEEL and Section \ref{sec:convergence} shows its convergence. Simulation results are presented in Section \ref{sec:experiment} and this paper is concluded in Section \ref{sec:conclusion}.

%%%%%%%%%%%%%%%%%%%%%%%%%%%%%%%%%%%%%%%%%%
%%%%%%%%%%%%%% System Model %%%%%%%%%%%%%%
%%%%%%%%%%%%%%%%%%%%%%%%%%%%%%%%%%%%%%%%%%
\section{System Model}\label{sec:system}

We consider an SD-FEEL system consisting of $C$ client nodes (denoted as set $\mathcal{C}$) and $D$ edge servers (denoted as set $\mathcal{D}$). Each client node is associated with an edge server, and it is assumed that each edge server has at least one associated client node. The client nodes thus form $D$ edge clusters and each edge cluster is coordinated by an edge server, acting as the parameter server (PS) in FL. Specifically, edge cluster $d$ comprises edge server $d$ and a subset of the client nodes, denoted as $\mathcal{C}_d$ (i.e., $\cup_{d\in\mathcal{D}}{\mathcal{C}_d}=\mathcal{C}$). The edge servers are partially inter-connected via high-speed links, and the connectivity can be captured by a connectivity matrix $\mathbf{G} \!\triangleq\!  \{g_{d,j}\} \!\in\! \{0,1\}^{D\!\times\! D}$, where $g_{d,j}=1$ means edge servers $d$ and $j$ are connected and otherwise, $g_{d,j}=0$. Denote $\mathcal{N}_d\triangleq \{j\in\mathcal{D} | g_{d,j} \!=\! 1\}$ as the set of one-hop neighbors of edge server $d$. 

Each client node has a set of local training data, denoted as $\mathcal{S}_i = \{\bm{\xi}_{j}^{(i)}\}_{j=1}^{|\mathcal{S}_i|}, i\in \mathcal{C}$, where $\bm{\xi}_{j}^{(i)}$ is the $j$-th data sample at client node $i$. The collection of data samples at the set of client nodes $\mathcal{C}_{d}$ is denoted as $\Tilde{\mathcal{S}}_d$, and the training data at all the client nodes is denoted as $\mathcal{S}$. We define $\hat{m}_i \triangleq \frac{|\mathcal{S}_i|}{|\Tilde{\mathcal{S}}_d|}$, $m_i \triangleq \frac{|\mathcal{S}_i|}{|\mathcal{S}|}$, and $\Tilde{m}_d \triangleq \frac{|\Tilde{\mathcal{S}}_d|}{|\mathcal{S}|}$. The client nodes collaborate to train a shared DL model $\bm{w}\in\mathbb{R}^M$ without disclosing their local data, where $M$ is the number of trainable parameters. The loss of data sample $\bm{\xi}$ is defined as a function of the model parameters, denoted as $f(\bm{\xi}; \bm{w})$, which can be, for example, the categorical cross-entropy between the predicted label and the ground truth for classification tasks. The objective of SD-FEEL is to minimize the loss over all the training data samples by optimizing $\bm{w}$, i.e., $\min_{\bm{w}\in\mathbb{R}^{M}} \left\{ F(\bm{w})\triangleq \sum_{i\in\mathcal{C}} m_i F_i(\bm{w}) \right\}$, where $F_i(\bm{w}) \triangleq \frac{1}{|\mathcal{S}_i|} \sum_{j\in\mathcal{S}_i} f(\bm{\xi}_j^{(i)}; \bm{w})$.

Same as \cite{sun2021semi}, we assume the local data across different client nodes are non-independent and identically distributed (non-IID). Besides, the real-life scenarios where client nodes have diversified computational resources will be investigated. We denote the computation speed at client node $i$ as $h_i$, which is in the unit of floating point operations per second (FLOPS). In other words, the number of local training epochs that client node $i$ can perform in a given amount of time is proportional to $h_i$ \cite{wang2020tackling}. To characterize the device heterogeneity, we introduce the \emph{heterogeneity gap} defined as $H\triangleq \max_{i,j\in\mathcal{C}}\frac{h_i}{h_j}$, of which, a larger value implies a higher degree of computational resource imbalance among the client nodes. Note that when $H\!=\!1$, all client nodes have identical computational speed, and synchronous training is viable for SD-FEEL systems \cite{sun2021semi,castiglia2020multi}. However, with $H \gg 1$, assigning equal number of local training epochs for all the client nodes as the operations in synchronous SD-FEEL will degrade the training efficiency considerably, since the fast client nodes have to wait until all the straggling client nodes complete their local model updates. In the next section, we propose an asynchronous training algorithm for SD-FEEL, where the client nodes perform different numbers of local epochs in each training round according to their computation speeds.

%%%%%%%%%%%%%%%%%%%%%%%%%%%%%%%%%%%%%%%%%%
%%%%%%%%%%%%%% Training Part %%%%%%%%%%%%%
%%%%%%%%%%%%%%%%%%%%%%%%%%%%%%%%%%%%%%%%%%
\section{Asynchronous Training in SD-FEEL}\label{sec:training}
In SD-FEEL, each training iteration includes three key steps, namely local model update, intra-cluster model aggregation, and inter-cluster model aggregation. To reduce the training latency, we design an asynchronous training algorithm for SD-FEEL that allows each edge cluster proceed to the next training iteration once it completes model sharing with the neighboring edge clusters. 
We denote the iteration counter as $k$, which gives the total number of training iterations that have been finished by all the edge clusters.
Fig. \ref{fig:flow} illustrates the training process of asynchronous SD-FEEL, and the system operations of the three key steps in each iteration are detailed as follows. 
\begin{figure*}[t]
\setlength\belowcaptionskip{-9pt}
    \centering
    \includegraphics[width=\linewidth]{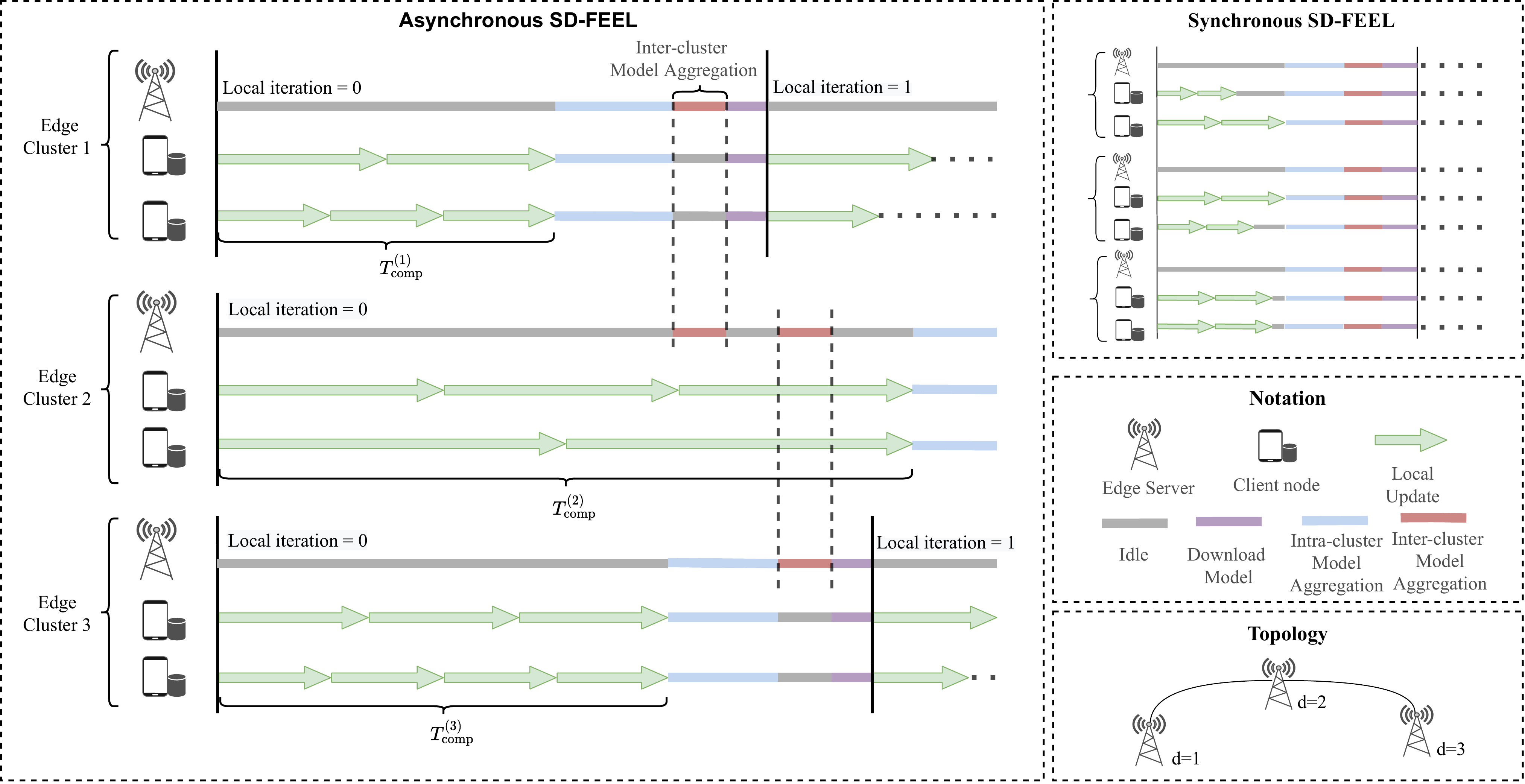}
    \caption{An illustration of asynchronous (left) and synchronous (top-right) SD-FEEL. In asynchronous SD-FEEL, the client nodes in edge cluster $d\in\{1,2,3\}$ perform local model updates for a duration of $T_\text{comp}^{(d)}$ before uploading the model updates to the associated edge server. The edge server then aggregates the received models from the client nodes in its own cluster and shares the aggregated model with its one-hop neighbors. In synchronous SD-FEEL, the client nodes are required to perform the same number of local epochs, where the fast client nodes need to stay idle until all the client nodes complete their local training before intra-cluster and inter-cluster model aggregations.}
    \label{fig:flow}
\end{figure*}

\subsubsection{Local Model Update} We assume edge server $d$ presets a deadline $T_\text{comp}^{(d)}$ for local model updates in each training iteration according to computational resources of its associated client nodes $\mathcal{C}_d$. While it is beyond the scope of this paper, we note that the value of $T_{\text{comp}}^{(d)}$ should ensure sufficient local training and avoid high model divergence \cite{wang2019adaptive}. Within the duration of $T_\text{comp}^{(d)}$, client node $i$ performs $\tau_i \!=\! \beta h_i$ epochs of mini-batch stochastic gradient descent (SGD), where $\beta$ is related to the complexity of training task and the batch size. Denote the model on client node $i$ at the beginning of local training epoch $l$ in the $k$-th global iteration as $\bm{w}_{k,l}^{(i)}$. Thus, we have
\begin{equation}
\bm{w}_{k,l+1}^{(i)} \!\!\leftarrow\! \bm{w}_{k,l}^{(i)} \!-\! \eta g(\bm{\xi}_{k,l}^{(i)}; \bm{w}_{k,l}^{(i)}), l \!\in \!\! \left\{0,1,\!\dots\!,\tau_i\!-\!1\right\}, i \!\in\! \mathcal{C},
\label{eq:local-update}
\end{equation}
where $\eta$ denotes the learning rate, and $g(\bm{\xi}_{k,l}^{(i)}; \bm{w}_{k,l}^{(i)})$ is the gradient computed on a randomly-sampled batch of local data $\bm{\xi}_{k,l}^{(i)}$. Since the numbers of local epochs vary among different client nodes, the local updates to be uploaded are normalized by $\tau_i$ in order to reduce model bias towards the fast client nodes \cite{wang2020tackling}, as follows:
% \vspace{-5pt}
\begin{equation}
    \bm{\Delta}_{k}^{(i)} \!\triangleq\! \frac{1}{\tau_i} \! \left(\bm{w}_{k,\tau_i}^{(i)} - \bm{w}_{k,0}^{(i)}\right) \!=\! - \frac{\eta}{\tau_i} \! \sum_{l=0}^{\tau_i-1} \! g(\bm{\xi}_{k,l}^{(i)}; \bm{w}_{k,l}^{(i)}), i\in \mathcal{C}. \vspace{-2pt}
\label{eq:delta}
\end{equation}

\subsubsection{Intra-cluster Model Aggregation} Once the deadline $T_\text{comp}^{(d)}$ of each training iteration arrives, the client nodes in edge cluster $d$ upload the normalized local updates $\bm{\Delta}_{k}^{(i)}$ as obtained in \eqref{eq:delta} to the associated edge server. Denote the model maintained by edge server $d$ at the beginning of the $k$-th global iteration as $\bm{y}_{k}^{(d)}$. The received local updates $\{\bm{\Delta}_{k}^{(i)}\}$'s are first weighted averaged using the weighting factors $\{\hat{m}_i\}$'s before being added to $\bm{y}_{k}^{(d)}$ for the implementation of gradient descent, which can be expressed as follows:
\vspace{-2pt}
\begin{equation}
    \bm{\hat{y}}_{k}^{(d)} 
    \leftarrow \bm{y}_{k}^{(d)} + \overline{\tau}_d \sum_{i\in\mathcal{C}_d} \hat{m}_i \bm{\Delta}_{k}^{(i)}, d \in \mathcal{D}.
    \vspace{-2pt}
\label{eq:intra}
\end{equation}
where $\overline{\tau}_d \triangleq \sum_{i\in \mathcal{C}_d} \hat{m}_i \tau_i$ is the weighted average of the numbers of local epochs completed by the client nodes.

\subsubsection{Inter-cluster Model Aggregation} After intra-cluster model aggregation, edge server $d$ shares the most updated model $\bm{\hat{y}}_{k}^{(d)}$ with its neighboring edge servers $\mathcal{N}_d$. Accordingly, the models maintained by these edge servers are updated as follows:
\vspace{-3pt}
\begin{equation}
    \bm{y}_{k}^{(j)} \leftarrow \sum_{j^\prime \in \mathcal{N}_{j}\cup\{j\}} p_{k}^{j^\prime,j} \bm{\hat{y}}_{k}^{(j^\prime)}, j \in \mathcal{N}_{d}\cup\{d\},
    \label{eq:inter}
    \vspace{-5pt}
\end{equation}

\noindent where $\mathbf{P}_{k} \triangleq \{p_{k}^{j^\prime,j}\}$ denotes the mixing matrix that may possibly change over different training iterations. It is worth noting that in synchronous SD-FEEL, $\mathbf{P}_{k}$ is a constant matrix over time determined by the connectivity among the edge servers \cite{castiglia2020multi}. However, in asynchronous training, when an edge server initiates inter-cluster model aggregation in the $k$-th global iteration, client nodes in neighboring edge cluster $j$ may be still training on models received in a previous global iteration $k^\prime(j)<k$, which are less valuable to the learning performance. Thus, we design a staleness-aware mixing matrix \cite{xie2019asynchronous}, of which, each element is non-increasing with the staleness $\delta_k^{(j)} \triangleq k - k^\prime(j)$ as shown in the following expression:
\begin{equation}
    p_{k}^{i,j} = \left\{
        \begin{array}{ll}
        \frac{\psi(\delta_k^{(j)})}{\Psi_j} & \text{if} \; j=d \; \text{and} \; i\in\mathcal{N}_{d}\cup\{d\}, \\
        p_{k}^{j,i} \; & \text{if} \; j \in\mathcal{N}_{d} \; \text{and} \; i=d,\\
        1-p_{k}^{d,j} \; & \text{if} \; j \; \in\mathcal{N}_{d} \; \text{and} \; i=j ,\\
        1 \; & \text{if} \; j \; \notin \mathcal{N}_{d}\cup\{d\} \; \text{and} \; i=j ,\\
        0, & \text{otherwise,}
        \end{array}
        \right.
\label{eq:pk}
\end{equation}
where $\psi(x)$ is a general non-increasing function of $x$ and $\Psi_k^{(j)} \triangleq \sum_{i\in\mathcal{N}_{j}\cup\{j\}} \psi(\delta_k^{(i)})$. The inter-cluster aggregated model $\bm{y}_{k}^{(d)}$ is then broadcasted to the client nodes in set $\mathcal{C}_d$, i.e.,
\vspace{-3pt}
\begin{equation}
    \bm{w}_{k+1,0}^{(i)} \leftarrow \bm{y}_{k}^{(d)}, i \in \mathcal{C}_d.
    \label{eq:broadcast}
    \vspace{-3pt}
\end{equation}

The above steps repeat until timeout or the values of local loss at all the client nodes cannot be further reduced. After that, the system enters a consensus phase where the edge servers exchange and aggregate models with their neighboring clusters. It was shown in \cite{sun2021semi} that model consensus among the edge servers can be achieved after sufficient rounds of such operations, i.e., the output model of asynchronous train is given as $\sum_{d\in\mathcal{D}} \Tilde{m}_d \bm{y}_{k}^{(d)}$, where $k$ is the global iteration index when the system enters the consensus phase. Note that as the consensus phase takes place only once, it shall induce negligible extra overhead. Details of the asynchronous training procedures for SD-FEEL are summarized in Algorithm \ref{alg-1}.
%%%%%%%%%%% Algorithm %%%%%%%%%%%
\begin{algorithm}[t]
\caption{Asynchronous Training for SD-FEEL} \label{alg-1}
% input & output
\SetKwInput{KwInput}{Input} 
\SetKwInput{KwOutput}{Output}
\KwInput{The randomly initialized model $\bm{y}_{0}$}
\KwOutput{The final model $\sum_{d\in\mathcal{D}} \Tilde{m}_d \bm{y}_{k}^{(d)}$}
%%%%%%%%% Edge Server
\SetKwFunction{Fupdate}{Update}
\SetKwProg{myproc}{Edge Server}{:}{}
  \myproc{}{
  Initialize all edge servers with the same model (i.e., $\bm{y}_{0}^{(d)} \!=\! \bm{y}_{0}, \, \forall d\in\mathcal{D}$), set $k \!=\! 0$\;
  \For{each edge server $d\in \mathcal{D}$ in parallel}{
    \Repeat{timeout or the values of the local loss functions cannot be further reduced}{
        \For{each client node $i\in \mathcal{C}_d$ in parallel}{
        $\bm{\Delta}_{k}^{(i)} \leftarrow$ \Fupdate{$\bm{y}_k^{(d)}$}\;
        }
        Perform intra-cluster model aggregation according to \eqref{eq:intra}\;
        Exchange the most updated model $\bm{\hat{y}}_{k}^{(d)}$ with its one-hop neighbors $\mathcal{N}_{d}$\;
        Perform inter-cluster model aggregation according to \eqref{eq:inter}\;
        Broadcast the most updated model $\bm{y}_{k}^{(d)}$ to the associated client nodes $i\in \mathcal{C}_d$ according to \eqref{eq:broadcast}\;
        Update $k \leftarrow k + 1$\;
    }
  }
  Enter the consensus phase\;
  \KwRet $\sum_{d\in\mathcal{D}} \Tilde{m}_d \bm{y}_{k}^{(d)}$\;
}
%%%%%%%%% Client Node
\SetKwProg{myprocc}{Client Node}{:}{}
  \myprocc{\Fupdate{$\bm{y}_k^{(d)}$}}{
  Initialize the local model, i.e., $ \bm{w}_{k,0}^{(i)} = \bm{y}_k^{(d)}$\;
  \Repeat{$T_\mathrm{comp}^{(d)}$ runs out}{
    Perform local training according to \eqref{eq:local-update}\;
  }
  Compute $\bm{\Delta}_{k}^{(i)}$ according to \eqref{eq:delta}\;
  \KwRet $\bm{\Delta}_{k}^{(i)}$ to edge server $d$\;}
\end{algorithm}
% \setlength{\textfloatsep}{5pt}

%%%%%%%%% Convergence Analysis %%%%%%%%%%%
\section{Convergence Analysis}\label{sec:convergence}

% \subsection{Assumptions}
To advance the convergence analysis of the proposed asynchronous training algorithm for SD-FEEL, we make the following assumptions on the loss functions that are commonly adopted in existing FL literature \cite{sun2021semi,lian2018asynchronous,liu2020client,wang2019adaptive}.
\begin{assumption}\label{ass-1}
For all $i\in \mathcal{C}$ and $\bm{w},\bm{w}^\prime \in\mathbb{R}^M$, we assume:
\begin{itemize}[leftmargin=1em]
    \item The local loss function is $L$-smooth, i.e.,
        \begin{equation}
        % \left\| \nabla\! F_i(\!\bm{w}_1\!) \!-\! \nabla\! F_i(\!\bm{w}_2\!) \right\| \!\leq\! L\!\left\| \bm{w}_1\!-\!\bm{w}_2 \right\|, \! \forall \bm{w}_1,\!\bm{w}_2\!\in\!\mathbb{R}^{\!M}.
        \left\| \nabla F_i(\bm{w}) - \nabla F_i(\bm{w}^\prime) \right\| \leq L \left\| \bm{w} - \bm{w}^\prime \right\|.
        \label{eq-smooth}
        \end{equation}
    \item The mini-batch gradient $g_i\left(\bm{\xi}; \bm{w}\right)$ is unbiased, i.e.,
        \begin{equation}
            \mathbb{E}_{\bm{\xi}|\bm{w}} [g_i(\bm{\xi};\bm{w})] = \nabla F_i(\bm{w}),
        \label{eq-gradient}
        \end{equation}
        and there exists $\sigma>0$ such that
        \begin{equation}
            \mathbb{E}_{\bm{\xi}|\bm{w}} \left[\left\| g_i(\bm{\xi};\bm{w}) - \nabla F_i(\bm{w}) \right\|^2\right] \leq \sigma^2.
        \label{eq-variance}
        \end{equation}
    \item There exists $\kappa>0$ such that
        \begin{equation}
            \left\| \nabla F_i(\bm{w}) - \nabla F(\bm{w}) \right\| \leq \kappa,
        \label{kappa}
        \end{equation}
        where $\kappa$ measures the non-IIDness of the training data across different client nodes.
\end{itemize}
\end{assumption}

%%%%%%%%%%%%%%%%%% Proof of Sketch
% \subsection{Convergence Results}
We define an auxiliary global model at the $k$-th training iteration as a weighted average of the models maintained by the edge servers, i.e., $\bm{\overline{y}}_{k} \triangleq \sum_{d\in\mathcal{D}} \Tilde{m}_d \bm{y}_k^{(d)}$. The evolution of $\bm{\overline{y}}_{k}$ can be expressed as follows:
\begin{equation}
    \bm{\overline{y}}_{k+1} = \bm{\overline{y}}_{k} - \eta \mathbf{\hat{G}}_k \mathbf{\Lambda} \bm{\Tilde{m}}^\mathrm{T},
    \label{eq:server-change}
\end{equation}
where $\mathbf{\hat{G}}_k \triangleq \left[ \sum_{i\in\mathcal{C}_d} \frac{\hat{m}_i}{\tau_i} \sum_{l=0}^{\tau_i-1} g(\bm{\xi}_{k,l}^{(i)};\bm{w}_{k,l}^{(i)}) \right]_{d\in\mathcal{D}} \!\in\! \mathbb{R}^{M\!\times\! D}$, $\mathbf{\Lambda} \triangleq \text{diag}(\overline{\tau}_1, \overline{\tau}_2,\dots, \overline{\tau}_D)$, and $\bm{\Tilde{m}} \triangleq \left[ \Tilde{m}_d \right]_{d\in\mathcal{D}}$. Once any edge cluster completes an iteration, $k$ increases, while other edge clusters are utilizing stale models for local training as aforementioned. Thus, we define $\bm{a}_{\Tilde{k}}^{(d)} \triangleq \bm{a}_{k-\delta_k^{(d)}}^{(d)}$ (respectively $\bm{a}_{\Tilde{k}}^{(i)} \triangleq \bm{a}_{k-\delta_k^{(d)}}^{(i)}, i \!\in\! \mathcal{C}_d$) as the delayed model or gradient at the $d$-th edge server (respectively $i$-th client node) in the $k$-th iteration. A large value of $\delta_k^{(d)}$ implies edge cluster $d$ is utilizing an outdated model which has less contribution to the global training and hinders the convergence.
The following lemma shows that through the whole training process $\delta_k^{(d)}$ is upper bounded.
\begin{lemma}\label{lemma-delta}
There exists a constant $\delta_\text{max}$ such that $\delta_k^{(d)} \leq \delta_\text{max}, \forall k \in \mathbb{N}, d\in\mathcal{D}$.
\end{lemma}
\begin{proofsketch}
During one training iteration of the slowest edge cluster, any other edge clusters can finish at most $H$ iterations. 
The iteration gap is bounded by the sum of training iterations other edge servers have triggered.
\end{proofsketch}

We are now to bound the expected change of loss functions in consecutive iterations as shown in the following lemma.
%%%%%%%%%% lemma 1 %%%%%%%%%%%%
\begin{lemma}\label{lemma-1}
The expected change of the global loss function in consecutive iterations is bounded as follows:
% \end{lemma}
% \begin{lemma*}
\begin{equation}
\begin{split}
    & \mathbb{E}F(\bm{\overline{y}}_{k+1}) - \mathbb{E} F(\bm{\overline{y}}_{k})
    \leq - \frac{1}{2} \eta\tau_{\mathrm{min}} \mathbb{E} \left\| \nabla F(\bm{\overline{y}}_k) \right\|^2 \\
    & - \frac{\eta}{2} (\tau_{\mathrm{min}} - \eta L  \tau_{\mathrm{max}}^2) J_k + \frac{1}{2} \eta^2 L \tau_{\mathrm{max}} H \sum_{i\in\mathcal{C}} m_i^2 \sigma^2 \\
    & +\frac{1}{2} \eta \tau_{\mathrm{min}} \underbrace{\mathbb{E} \left\| \nabla F(\bm{\overline{y}}_k) -  \nabla\mathbf{\hat{F}}_{\Tilde{k}} \bm{\Tilde{m}}^\mathrm{T} \right\|^2}_{\mathcal{E}_k},
    % \vspace{-1pt}
\label{eq:loss-change}
\end{split}
\end{equation}
where $\tau_{\mathrm{min}}=\min_{i\in\mathcal{C}} \tau_i$, $\tau_{\mathrm{max}}=\max_{i\in\mathcal{C}} \tau_i$,  $J_k \!\triangleq\! \mathbb{E} \| \nabla\mathbf{\hat{F}}_{\Tilde{k}} \bm{\Tilde{m}}^\mathrm{T} \|^2$, and $\nabla\mathbf{\hat{F}}_{\Tilde{k}} \!\triangleq\! [ \sum_{i\in\mathcal{C}_d} \frac{\hat{m}}{\tau_i} \sum_{l=0}^{\tau_i-1} \nabla F_i (\bm{w}_{\Tilde{k},l}^{(i)}) ]_{d\in\mathcal{D}}$.
\end{lemma}
\begin{proof}
The proof can be obtained by plugging the right-hand side (RHS) of \eqref{eq:server-change} into the first-order Taylor expansion of $F(\bm{\overline{y}}_{k+1})$.
Then we apply the assumptions in \eqref{eq-smooth} and \eqref{eq-variance} to conclude the proof.
\end{proof}

The term $\mathcal{E}_k$ in \eqref{eq:loss-change} measures the degree to which the gradients collected from client nodes (i.e., $\nabla\mathbf{\hat{F}}_{\Tilde{k}} \bm{\Tilde{m}}^\mathrm{T}$) deviate from the desired gradient of global model (i.e., $\nabla F(\bm{\overline{y}}_k)$). In the following lemma, we derive an upper bound for $\mathcal{E}_k$.
%%%%%%%%%% lemma 2 %%%%%%%%%%%%
\begin{lemma}\label{lemma-2}
With Assumption \ref{ass-1}, we have
\begin{equation}
\begin{split}
    &\quad \frac{1}{K} \! \sum_{k=0}^{K-1} \mathcal{E}_k
    \leq A (\tau_\mathrm{max},\delta_\mathrm{max},H) \sigma^2 \\
    & + B (\tau_\mathrm{max},\delta_\mathrm{max},H) \kappa^2 
    + C (\tau_\mathrm{max},\delta_\mathrm{max}) J_k,
\end{split}
\label{eq-lemma}
\end{equation}
where 
\begin{equation*}
\begin{split}
    & A (\tau_\mathrm{max},\delta_\mathrm{max},H) \triangleq 4 \eta^2 L^2 \delta_\mathrm{max}^2 \tau_{\mathrm{max}} H U_4 +\\
    &\frac{4\eta^2L^2 (\tau_\mathrm{max}-1)}{1-2\eta^2L^2 U_2} % \vspace{-1pt}
    \!+\! 8\eta^2 L^2 \tau_{\mathrm{max}} H U_3  \frac{1}{K} \!\sum_{k=0}^{K-1} \sum_{s=0}^{k-1} \rho_{s,k-1}^2, \\ %\vspace{-1pt}
    & B (\tau_\mathrm{max},\delta_\mathrm{max},H) \triangleq 8 \eta^2 L^2 \delta_\mathrm{max}^2 \tau_{\mathrm{max}} H U_4 +\\
    &\!\frac{24\eta^2L^2U_2}{1 \!-\! 2\eta^2L^2 U_2}
    \!+\! 16 \eta^2 L^2 \tau_{\mathrm{max}} H U_3 \frac{1}{K} \!\sum_{k=0}^{K-1} \bigg( \sum_{s=0}^{k-1} \rho_{s,k-1} \bigg)^2, \\ %\vspace{-1pt}
    & C (\tau_\mathrm{max},\delta_\mathrm{max},H) \triangleq 8 \eta^2 L^2 \delta_\mathrm{max}^2 \tau_\mathrm{max} U_4 \\ \vspace{-1pt}
    &+ 16 \eta^2 L^2 \tau_\mathrm{max}^2 U_3 \frac{1}{K} \!\sum_{k=s+1}^{K-1} \rho_{s,k-1} \bigg( \sum_{l=0}^{k-1} \rho_{l,k-1} \bigg), \\ % \vspace{-1pt}
    & U_2 \!\triangleq\! \tau_\mathrm{max} (\tau_\mathrm{max}\!-\!1), U_3 \!\triangleq\! \frac{1\!+\!4\eta^2L^2 U_2}{1\!-\!2\eta^2L^2 U_2}, U_4 \!\triangleq\! \frac{1 \!+\! 22\eta^2L^2 U_2}{1\!-\!2\eta^2L^2 U_2}.
\end{split}
\end{equation*}
\end{lemma}
\begin{proofsketch}
We first derive an upper bound for $\mathcal{E}_k$ by applying \eqref{eq-smooth} and Jensen's inequality (i.e., $\|\bm{a}+\bm{b} \|^2 \leq 2\| \bm{a} \|^2 + 2\| \bm{b} \|^2, \forall \bm{a},\bm{b} \in\mathbb{R}^d$) as follows:
\begin{equation}
\begin{split}
    & \mathcal{E}_k \leq 2 L^2 \underbrace{\left\| \bm{\overline{y}}_k - \bm{\overline{y}}_{\Tilde{k}} \right\|^2}_{\mathcal{E}_{k,1}} + 4L^2 \underbrace{\sum_{d\in \mathcal{D}} \Tilde{m}_d \mathbb{E}\left\| \bm{\overline{y}}_{\Tilde{k}} - \bm{y}_{\Tilde{k}}^{(d)} \right\|^2}_{\mathcal{E}_{k,2}} \\ %\vspace{-2pt}
    & + 4L^2 \sum_{d\in \mathcal{D}} \Tilde{m}_d \underbrace{ \sum_{i\in \mathcal{C}_d} \frac{\hat{m}_i}{\tau_i} \sum_{l=0}^{\tau_i-1} \mathbb{E}\left\| \bm{y}_{\Tilde{k}}^{(d)} - \bm{w}_{\Tilde{k},l}^{(i)} \right\|^2}_{\mathcal{E}_{k,3}}. 
    % \vspace{-2pt}
\end{split}
\label{eq:respective}
\end{equation}
We respectively bound the three terms in the RHS of \eqref{eq:respective} by bounding the accumulated gradients computed at the client nodes.
The term $\mathcal{E}_{k,1}$ is caused by the iteration gap, upper bounded by using Lemma \ref{lemma-delta}.
The term $\mathcal{E}_{k,2}$ measures the consensus error among edge clusters, which can be characterized by the accumulated divergence between the mixing matrix and the expected average, denoted as $\rho_{s,k-1} \triangleq \left\| \prod_{l=s}^{k-1} \mathbf{P}_l - \bm{\Tilde{m}}\mathbf{1} \right\|_\text{op}$, where $\left\| \cdot \right\|_\text{op}$ is the operator norm of the matrix. Denote the maximum value of the second largest eigenvalue of any $\mathbf{P}_k$ as $\rho_\text{max}$, then $\sum_{s=0}^{k-1} \rho_{s,k-1} \!\leq\! \frac{1}{1-\rho_\text{max}}$ holds.
The term $\mathcal{E}_{k,3}$ measures the model divergence within any edge cluster $d$, which is introduced by biased local training.
The details are omitted due to space limitation.
\end{proofsketch}

%%%%%%%%%%%%%%%%%% Main Results
% \subsection{Main Results}
With Lemma \ref{lemma-1} and Lemma \ref{lemma-2}, we show the convergence of asynchronous SD-FEEL in the following theorem.
\begin{theorem} \label{thm}
With Assumption \ref{ass-1}, if the learning rate $\eta$ satisfies
    \begin{equation}
    1 - \eta L  H \tau_\mathrm{max} - C  (\tau_\mathrm{max},\delta_\mathrm{max},H) \geq 0, 1 - 2\eta^2L^2 U_2 > 0,
    \label{eq:lr}
    \end{equation}
    we have
    \begin{align}
    % \begin{split}
    &\quad \frac{1}{K} \sum_{k=0}^{K-1} \mathbb{E} \left\| \nabla F(\bm{\overline{y}}_k) \right\|^2 \nonumber\\
    &\leq\frac{2 [\mathbb{E}F(\bm{\overline{y}}_{0}) - \mathbb{E} F(\bm{\overline{y}}_{K})]}{\eta \tau_\mathrm{min} U_1 K}
    + \frac{1}{U_1} \eta L H^2 \sum_{i\in\mathcal{C}} m_i^2 \sigma^2 \quad
    % \end{split}
    \label{eq:thm}
    \end{align}
    \vspace{-7pt}
    \begin{equation*}
        + A ( \tau_\mathrm{max},\delta_\mathrm{max},H) \frac{\sigma^2}{U_1} + B (\tau_\mathrm{max},\delta_\mathrm{max},H) \frac{\kappa^2}{U_1},
    \end{equation*}
    where $U_1 \!\triangleq\! \frac{1-14\eta^2L^2 U_2}{1-2\eta^2L^2 U_2}$.
\end{theorem}
\begin{proof}
We sum up both sides of \eqref{eq:loss-change} over $k\!=\!0,1,\dots,K\!-\!1$, divide them by $K$, and apply \eqref{eq-lemma} to the RHS. By rearranging terms and choosing the learning rate in \eqref{eq:lr}, we complete the proof.
\end{proof}

\begin{remark}\label{rm:lr}
If we choose the learning rate as $\eta=\mathcal{O} \left(\frac{1}{L\sqrt{K}} \right)$, the first two terms in the RHS of \eqref{eq:thm} dominate and decreases at a speed of $\mathcal{O}(\frac{1}{\sqrt{K}})$, i.e., $K \!\rightarrow\! \infty$, the RHS of \eqref{eq:thm} approaches zero, which ensures the convergence.
\end{remark}

\begin{remark}\label{rm:async}
When $H\!=\!1$ and $\delta_\text{max}\!=\!0$, the result in Theorem \ref{thm} reduces to the synchronous case \cite{sun2021semi}. Compared with this case, the RHS of \eqref{eq:thm} incorporates additional terms that increases with the heterogeneity gap $H$ and the iteration gap $\delta_\text{max}$, implying that SD-FEEL with larger device heterogeneity leads to slower convergence. Besides, severe data heterogeneity (i.e., large $\kappa$) exacerbates the error in the RHS of \eqref{eq:thm} and slows down the convergence as in synchronous training.
\end{remark}

%%%%%%%%% Section 4: experiments
\section{Simulation Results}\label{sec:experiment}
%%%%%%%%% Section 4.A
\subsection{Setup}
An SD-FEEL system with 30 client nodes being divided into six edge clusters is considered in our simulations, and each edge server is associated with five client nodes. Without loss of generality, we assume the connection among the edge servers forms a ring topology. 
We compare synchronous and asynchronous SD-FEEL on the CIFAR-10 image classification task, where the ResNet-18 model \cite{he2016deep} with $M=11,173,962$ parameters is trained. 
To simulate the data heterogeneity, we adopt a Dirichlet distribution $\text{Dir}_{30}(0.5)$ to sample the probabilities $\{p_{l,i}\}$'s, which is the proportion of the training samples of class $l$ to the $i$-th client node \cite{yurochkin2019bayesian}. Besides, the batch size and the learning rate in mini-batch SGD are set to be 10 and 0.001, respectively. As an example, we use $\psi(\delta_k^{(j)})=\frac{1}{2 (\delta_k^{(j)}+1)}$ to calculate the mixing matrix for inter-cluster model aggregation.

The training latency for one iteration of edge cluster $d$ can be expressed as $T_\text{iter}^{(d)} = T_\text{comm}^\text{ct-sr} + T_\text{comm}^\text{sr-sr} + T_\text{comp}^{(d)}$, where $T_\text{comm}^\text{ct-sr}=\frac{M_{\rm{bit}}}{R_\text{comm}^\text{ct-sr}}$ and $T_\text{comm}^\text{sr-sr}=\frac{M_{\rm{bit}}}{R_\text{comm}^\text{sr-sr}}$ are respectively the delays for transmitting a model with $32M \,\text{bits}$ from the client node to the edge server and between two edge servers. Following the setting in our previous work \cite{sun2021semi}, the client nodes upload the updates to the edge server using a wireless channel and the transmission rate is given by $R^\text{ct-sr} = 5\,\text{Mbps}$. For the inter-server communication, it is via the high-speed links with the bandwidth of $10\,\text{Mbps}$. The average computation latency for one local epoch is formulated as $T_\text{comp}^\text{avg} = \frac{N_\text{MAC}}{C_\text{CPU}}$, where $N_\text{MAC} = 55.67\,\text{GFLOPs}$ is the required number of the floating-point operations (FLOPs), and $C_\text{CPU} = 1 \,\text{GFLOPS}$. The value of $T_\text{comp}^{(d)}$ in each cluster is respectively set to ensure at least $100$ mini-batches to be processed on client nodes.

%%%%%%%%% Section 4.B
\subsection{Results}\label{sec:result}
\begin{figure}[t]
\setlength\abovecaptionskip{-1pt}
\setlength\belowcaptionskip{-8pt}
    \centering
    \subfigure
    {\includegraphics[width=0.49\linewidth]{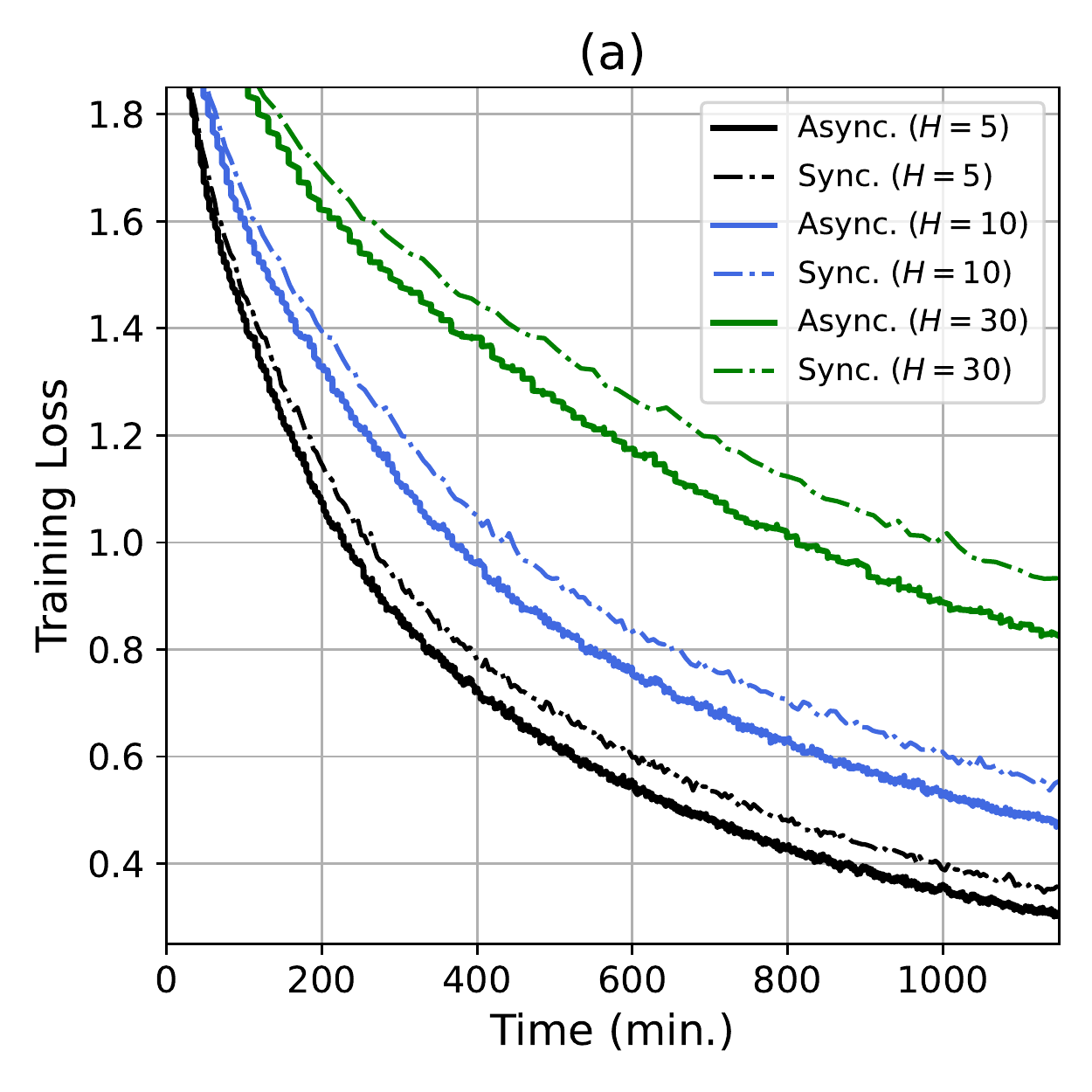}}
    \subfigure
    {\includegraphics[width=0.49\linewidth]{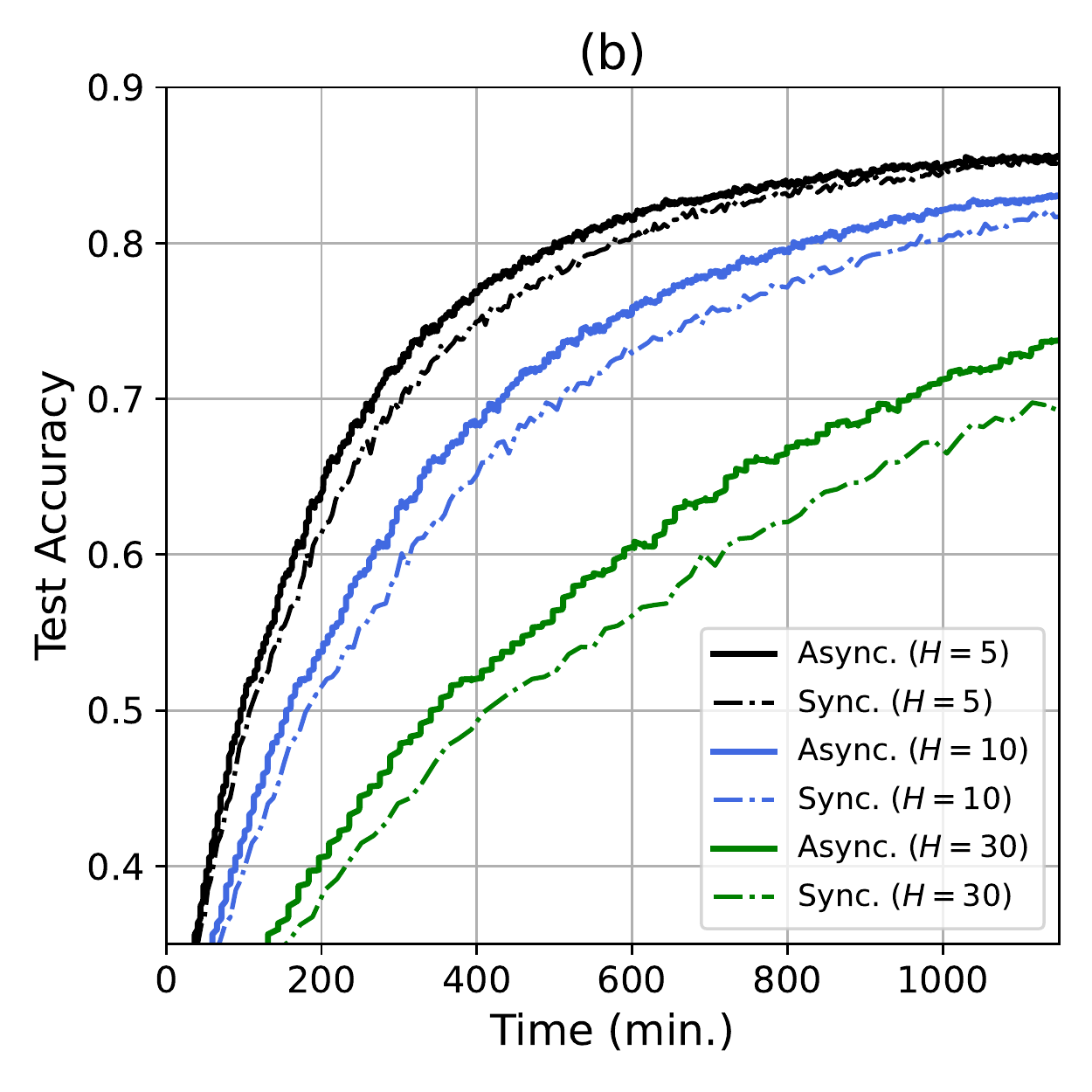}}
    \caption{(a) Training loss and (b) test accuracy over time with different heterogeneity gaps.}
    \label{fig:gap}
\end{figure}
Fig. \ref{fig:gap} shows the training loss and test accuracy over time with different degrees of device heterogeneity (i.e. $H=5$, $10$ and $30$). It is observed that in all the cases asynchronous SD-FEEL (denoted as \textit{Async.}) enjoys a faster convergence speed compared with synchronous training (denoted as \textit{Sync.}).
With a large value of $H$, slow client nodes have weaker computation capabilities and thus it takes longer for them to complete local training, which hinders the convergence. Nevertheless, asynchronous training effectively increases training efficiency by allowing faster client nodes to perform more local epochs and reducing their idle time.
According to Fig. \ref{fig:gap}(b), within a relatively short training time, asynchronous SD-FEEL obtains an improvement in the test accuracy compared with the baseline. However, if given sufficiently long training time such that the slow client nodes are utilized to a greater extent, synchronous SD-FEEL can reach similar test accuracies and may even perform better than asynchronous training algorithm.

%%%%%%%%% Section 5
\section{Conclusions}\label{sec:conclusion}
% \vspace{-3pt}
In this paper, we considered a practical scenario of SD-FEEL with device heterogeneity and designed an asynchronous training algorithm. We provided the convergence analysis and showed that the proposed asynchronous training algorithm secures notable improvement on convergence speed compared with synchronous training through numerical experiments. For future works, it is worth investigating how to determine the optimal local training time for each edge cluster. In addition, extending the proposed algorithm and analysis to SD-FEEL systems with dynamic computation and communication speeds at the client nodes is also necessary.

%%%%%%%%%%%% references
\bibliographystyle{IEEEtran} 
\bibliography{ref.bib}

\end{document}